\documentclass[11pt]{article}
\usepackage{mathrsfs}
\usepackage{CJK}
\usepackage{bbm}
\usepackage{graphicx}
\usepackage{amsmath,amssymb}
\usepackage{amsfonts}
\usepackage{amsthm}
\usepackage{verbatim}
\usepackage{color}
\usepackage{graphicx}
\usepackage{amssymb}
\usepackage{latexsym}
\large\normalsize

\newcommand{\be}{\begin{equation}}
\newcommand{\ee}{\end{equation}}
\newcommand{\bes}{\begin{equation}\begin{aligned}}
\newcommand{\ees}{\end{aligned}\end{equation}}
\newcommand{\ben}{\begin{equation}\nonumber\begin{aligned}}

 \newcommand{\R}{\mathbb{R}}

\newcommand{\N}{\mathbb{N}}

\renewcommand{\leq}{\leqslant}
\renewcommand{\geq}{\geqslant}

\newtheorem{thm}{Theorem}[section]
\newtheorem{lem}[thm]{Lemma}

\newtheorem*{main thm}{Main Theorem}

\numberwithin{equation}{section}
\begin{document}

\title{A sharp bound on RIC in generalized orthogonal matching pursuit}

\author{Wengu~Chen and Huanmin~Ge
\thanks{W. Chen is with Institute of Applied Physics and Computational Mathematics,
Beijing, 100088, China, e-mail: chenwg@iapcm.ac.cn.}
\thanks{H. Ge is with Graduate School, China Academy of Engineering Physics,
Beijing, 100088, China, e-mail:gehuanmin@163.com.}
\thanks{This work was supported by the NSF of China (Nos.11271050, 11371183)
.} }

\maketitle

\begin{abstract}
Generalized orthogonal matching pursuit (gOMP) algorithm has received much attention in recent years as a natural extension of orthogonal matching pursuit. It is used to recover  sparse signals in compressive sensing.
In this paper, a new bound is obtained for the exact reconstruction of every $K$-sparse signal via the gOMP algorithm in the noiseless case. That is, if the restricted isometry constant (RIC) $\delta_{NK+1}$ of the sensing matrix $A$ satisfies
\begin{eqnarray*}
\delta_{NK+1}<\frac{1}{\sqrt{\frac{K}{N}+1}},
\end{eqnarray*}
then the  gOMP can perfectly recover every $K$-sparse signal $x$ from $y=Ax$. Furthermore, the bound is proved to be sharp in the following sense. For any given positive integer $K$,  we construct a matrix $A$ with the RIC
\begin{eqnarray*}
\delta_{NK+1}=\frac{1}{\sqrt{\frac{K}{N}+1}}
\end{eqnarray*}
such that the gOMP may fail to recover some $K$-sparse signal $x$.
 In the noise case, an extra condition on the minimum magnitude of the nonzero components of every $K-$sparse signal combining with the above bound on RIC of the sensing matrix $A$ is sufficient to recover the true support of every $K$-sparse signal by the gOMP.
\end{abstract}


\section{Introduction}
It is well known that compressive sensing acquires sparse signals at a rate greatly below Nyquist rate. It has attracted growing attention in recent years \cite{1}-\cite{7}. The main aim of compressive sensing is  to reconstruct signal from inaccurate and incomplete measurements.
One consider the following compressive sensing model:
\begin{eqnarray*}
y=Ax+e,
\end{eqnarray*}
where $y\in \R^m$ is a measurement vector, the matrix $A\in\R^{m\times n} \ (m\ll n)$ is a sensing matrix, the vector $x\in\R^n$ is a unknown $K$-sparse signal $(K\ll n)$ and $e\in\R^m$ is a measurement error vector. The goal is to recover unknown signal $x$ based on $y$ and $A$.
 In this paper, denote by $A_i\ (i=1,2,\cdots,n)$ the $i$-th column of $A$ and all columns of $A$
are normalized, i.e., $\|A_i\|_2=1$ for $i=1,2,\cdots,n$.  Define
the support of the vector $x$ by $T=supp(x)=\{i|x_i\neq 0\}$ and the
size of its support  with $|T|=|supp(x)|$.  For a signal $x$, if
$|supp(x)|\leq K$, $x$ is called $K$-sparse.

For the recovery of the $K$-sparse signal $x$, the most intuitive approach is to solve the following optimization problem
\begin{eqnarray}\label{p1}
\min_{x}\|x\|_0\ \ \ \ \textrm{subject}\ \textrm{to}\  Ax-y\in \mathcal{B},
\end{eqnarray}
where $\|x\|_0$ denotes the $l_0$ norm of $x$, i.e., the number of nonzero coordinates, $\mathcal{B}$ is a bounded error set, i.e., $\mathcal{B}=\{e\in\R^m\ |\  \|e\|_2\leq\varepsilon \}.$  Particularly, in the  noiseless case, $\mathcal{B}=\{0\}$.
Unfortunately,  it is well-known that the above optimization problem is NP-hard. Therefore,
 researchers seek computationally efficient methods to approximate the sparse signal $x$, such as
   $l_1$ minimization \cite{8}, $l_p\,(0<p<1)$ minimization \cite{9}, greedy algorithm \cite{10} and so on.

 To ensure that the $K$-sparse solution is unique, we shall need the restricted isometry property (RIP) introduced by Cand\`{e}s and Tao  in \cite{8}. A matrix $A$ satisfies the restricted isometry property of order $K$ if there exists a constant $\delta_K$ such that
\begin{eqnarray}\label{p2}
(1-\delta_K)\|x\|_2^2\leq\|Ax\|_2^2\leq(1+\delta_K)\|x\|_2^2
\end{eqnarray}
holds for all $K-$sparse signals $x$. And the smallest constant $\delta_K$ is called as the restricted isometry constant (RIC). Cand\`{e}s and Tao also proposed that if $\delta_{2K}<1$, the above optimization problem has a unique $K$-sparse solution \cite{8}.  Cand\`{e}s showed that if $\delta_{2K}<\sqrt{2}-1$ then the above optimization problem \eqref{p1} is equivalent to the $l_1$ minimization problem in \cite{2}. Up to now, there are many results  improving the bound on the RIC such as \cite{4}, \cite{7} and \cite{15}-\cite{152}.

Recently,  there is  a family of iterative greedy algorithms which have attracted significant attention to recover sparse signals including orthogonal least square (OLS) \cite{aa}, orthogonal matching pursuit (OMP) \cite{11}, generalized orthogonal matching pursuit (gOMP) \cite{20},  regularized orthogonal matching pursuit (ROMP) \cite{12}, orthogonal multi-matching pursuit (OMMP) \cite {121},  stagewise orthogonal matching pursuit (StOMP) \cite{13}, subspace pursuit (SP) \cite{22} and compressive sampling matching pursuit (CoSaMP) \cite{12'}.

Specifically, OMP algorithm is one of the most effective algorithm in sparse signals recovery due to its  implementation simplicity and competitive recovery performance. In the noiseless case, many efforts
 have been made to find out sufficient conditions based on RIC for OMP to exactly
reconstruct every $K$-sparse signal $x$ within K iterations.
Davenport and Wakin demonstrated that OMP can recover exactly the
$K$-sparse signal $x$  under $\delta_{K+1}<\frac{1}{3\sqrt{K}}$
\cite{16}. Since then, there are many papers to improve
 the condition in \cite{161}-\cite{16'}. Recently, Mo improved the sufficient condition to
   $\delta_{K+1}<\frac{1}{\sqrt{K+1}}$, and proved this condition is sharp \cite{16'}. In the presence of noise,
    Shen  and Li proved that OMP can exactly recover the support of the $K$-sparse signal $x$
    under $\delta_{K+1}<\frac{1}{\sqrt{K}+3}$ and some assumption  on the minimum magnitude of the nonzero elements of $x$ in \cite{17}. Later,  these sufficient conditions on RIC upper bound and minimum magnitude of the nonzero elements of $K$-sparse signal $x$ have been improved in \cite{18} and \cite{19}.

Wang,  Kwon and Shim introduced generalized orthogonal matching pursuit \cite{20},
which is a natural extension  of OMP. It is well known that OMP algorithm only selects one correct index at each iteration. However  the gOMP algorithm selects $N\ (N\geq1)$ indices which contain at least one correct index from the  support of $x$ in each iteration. Therefore the number of iteration for the gOMP algorithm is much smaller  comparing with OMP algorithm.  Wang,  Kwon and Shim obtained that  a sufficient condition
\begin{eqnarray*}
\delta_{NK}<\frac{\sqrt{N}}{\sqrt{K}+3\sqrt{N}}
 \end{eqnarray*}
 can ensure the reconstruction of  any $K$-sparse signals  \cite{20}. Later,  Satpathi et al. improved the sufficient condition to $\delta_{NK}<\frac{\sqrt{N}}{\sqrt{K}+2\sqrt{N}}$ in \cite{21}. They also refined the bound further to
 \begin{eqnarray*}
 \delta_{NK+1}<\frac{\sqrt{N}}{\sqrt{K}+\sqrt{N}},
 \end{eqnarray*}
 which  is $\delta_{K+1}<\frac{1}{\sqrt{K}+1}$ of OMP in \cite{162} and  \cite{16''} for $N=1$.

Motivated by the mentioned papers, we further investigate the recovery of any $K$-sparse signals by the gOMP.
In this paper, we demonstrate  that the condition
\begin{eqnarray*}
\delta_{NK+1}<\frac{1}{\sqrt{\frac{K}{N}+1}}
\end{eqnarray*}
is sufficient to perfectly  reconstruct any $K$-sparse signals via the gOMP in the noiseless case.
As $N=1$, the sufficient condition is $\delta_{K+1}<\frac{1}{\sqrt{K+1}}$ which is a sharp bound for OMP  \cite{16'}. Moreover, for any given $K\in \N^+$, we construct a matrix $A$ satisfying
\begin{eqnarray*}
\delta_{NK+1}=\frac{1}{\sqrt{\frac{K}{N}+1}}
\end{eqnarray*}
such that the gOMP may fail to recover  some $K$-sparse signal $x$. That is, the above bound $\delta_{NK+1}<\frac{1}{\sqrt{\frac{K}{N}+1}}$ is sharp for the gOMP.
In noise case,  we also show $\delta_{NK+1}<\frac{1}{\sqrt{\frac{K}{N}+1}}$ together with a minimum magnitude of the nonzero elements of the $K$-sparse signal $x$ can ensure the reconstruction of the support of $x$ via the gOMP.

The frame of the gOMP is listed in the table $1$.

\begin{center} TABLE 1

The gOMP algorithm
\end{center}
\hrule
\textbf{Input}\ \ \ \ \ \ \ \ \ measurements $y\in\R^m$, sensing matrix $A\in\R^{m\times n}$, sparse level $K$, number of

\ \ \ \ \ \ \ \ \ \ \ \ \ \ indices for each selection $N$ $(N\leq K \ \textmd{and}\  N\leq\frac{m}{K})$.\\
\textbf{Initialize}\ \ \ \ iteration count $k=0$, residual vector $r^0=y$, estimated support set $\Lambda^0=\varnothing$.

\hrule
\textbf{While}\ \ \ \ \ $\|r^k\|_2>\epsilon$  and $k<\min\{K, \frac{m}{K}\}$ do \ $k=k+1$.

\ \ \ \ \ \ \ \ \ \ \ \ \ \ \ (Identification\ step)\ \ \ \ Select indices set $T^k$ corresponding to $N$ largest

\ \ \ \ \ \ \ \ \ \ \ \ \ \ \  \ \ \ \ \ \ \ \ \ \ \ \ \ \ \ \ \ \ \ \ \ \ \ \ \ \ \ \ \ \ (in  magnitude) in $A^{'}r^{k-1}$.

\ \ \ \ \ \ \ \ \ \ \ \ \ \ \ (Augmentation step)\ \ \ $\Lambda^k=\Lambda^{k-1}\cup T^k$.

\ \ \ \ \ \ \ \ \ \ \ \ \ \  \ (Estimation step)\ \ \ \ \ \ \ \ $\hat{x}_{\Lambda^k}=\arg\min\limits_{u}\|y-A_{\Lambda^k}u\|_2$.

\ \ \ \ \ \ \ \ \ \ \ \ \ \  \ (Residual Update step)\ \ \ $r^k=y-A_{\Lambda^k}\hat{x}_{\Lambda^k}$.

\textbf{End}

\textbf{Output}\ \ \ \  the estimated signal $\hat{x}=\arg\min\limits_{u:supp(u)=\Lambda^k}\|y-Au\|_2$
\hrule
\ \ \\

The rest of the paper is organized as follows. In Section \ref{2},
we give some notations and prove  some basic lemmas that will be used. The main
results and their proofs are given in Section \ref{3}.

\section{Notations and preliminaries}\label{preliminaries}\label{2}
 Throughout this paper, let $\Gamma$ be an index set and $\Gamma^c$ be the complementary set of $\Gamma$.  The standard notation $\|x\|_{\infty}=\max\limits_{i=1,2,\cdots,n}|x_i|$ denotes the $l_\infty-$norm of the vector $x\in\R^n$. $x_\Gamma\in\R^{|\Gamma|}$ denotes the vector composed of components of $x\in \R^n$ indexed by $i\in\Gamma$, i.e., $(x_\Gamma)_i=x_i \ (i\in\Gamma)$. Define $\tilde{x}_\Gamma\in\R^{n}$ by
 \begin{eqnarray*}
(\tilde{x}_\Gamma)_i= \left\{
   \begin{array}{ll}
     x_i, & \hbox{$i\in \Gamma$;} \\
     0, & \hbox{others,}
   \end{array}
 \right.
 \end{eqnarray*}
where $i=1,2,\cdots,n.$
Denote by $A_\Gamma$ a submatrix of $A$ corresponding to $\Gamma$ which consists of  all columns with index $i\in \Gamma$ of $A$ and the usual inner product of $\R^n$  with $\langle\cdot,\cdot\rangle$.
Let $e_i\in\R^n$ be the $i$-th coordinate unit vector.

Let $\alpha_N^{k+1}$ be the $N$-th largest correlation in magnitude between $r^k$ and $A_i\ (i\in(T\cup\Lambda^k)^c)$ and $\beta_1^{k+1}$ be the largest correlation in magnitude between $r^k$ and $A_i\ (i\in(T-\Lambda^k))$ in the $(k+1)$-th iteration of the gOMP algorithm. Let $W_{k+1}\subseteq(T\cup\Lambda^k)^c$  be the set of $N$ indices which correspond to
$N$ largest correlation in magnitude between $r^k$ and $A_i\ (i\in(T\cup\Lambda^k)^c)$.

$A_{\Lambda^k}^\dagger$ represents the pseudo-inverse of $A_{\Lambda^k}$ when $A_{\Lambda^k}$ is full column rank ($|\Lambda^k|\leq m$), i.e., $A_{\Lambda^k}^\dagger=(A_{\Lambda^k}^{'}A_{\Lambda^k})^{-1}A_{\Lambda^k}^{'}$. Moreover, $P_{\Lambda^k}=A_{\Lambda^k}A_{\Lambda^k}^\dagger$ and $P^\bot_{\Lambda^k}=I-P_{\Lambda^k}$ denote two orthogonal projection operators which  project a given vector orthogonally onto the spanned space by all columns of $A_{\Lambda^k}$ and onto its orthogonal complement respectively.

First, we recall the following lemma, that is, the monotonicity of the restricted isometry constant in \cite{8}, \cite{22}.
 \begin{lem}\label{lem1}For any $K_1\leq K_2$,
 if the sensing matrix $A$ satisfies the RIP of order $K_2$, then $\delta_{K_1}\leq\delta_{K_2}$.
 \end{lem}
 Next, we show  the  main lemma that is  very useful during  our analysis.
 \begin{lem}\label{lem2}
 For any $S,\ C>0$, let $t=\pm\frac{\sqrt{S+1}-1}{\sqrt{S}}$ and
 \begin{eqnarray}
 t_i=\left\{
       \begin{array}{ll}
         -\frac{C}{2}(1-t^2), & \hbox{$\langle Ax, Ae_i\rangle\geq 0$;} \\
         +\frac{C}{2}(1-t^2), & \hbox{$\langle Ax, Ae_i\rangle<0$,}
       \end{array}
     \right.
\end{eqnarray}
  where $i\in W\subseteq\{1,2,\cdots,n\}$ that is a nonempty subset. Then we have $t^2<1$ and
 \begin{eqnarray}\label{e1}
 \|A(x+\sum_{i\in W}t_ie_i)\|_2^2-\|A(t^2x-\sum_{i\in W}t_ie_i)\|_2^2=
 (1-t^4)\left(\langle Ax,Ax\rangle-C\sum_{i\in W}|\langle Ax, Ae_i\rangle|\right).\nonumber
 \end{eqnarray}
 \end{lem}
 \begin{proof}
 For $t=\pm\frac{\sqrt{S+1}-1}{\sqrt{S}}$, we have that
 \begin{eqnarray}\label{e2}
 t^2=\frac{(\sqrt{S+1}-1)^2}{S}=\frac{\sqrt{S+1}-1}{\sqrt{S+1}+1}<1.\nonumber
 \end{eqnarray}
  The result in the lemma is established by the following chain of equalities and the definition of $t_i\ (i\in W)$:
 \begin{eqnarray}
 &&\|A(x+\sum_{i\in W}t_ie_i)\|_2^2-\|A(t^2x-\sum_{i\in W}t_ie_i)\|_2^2\nonumber\\
 &&=\langle Ax,Ax\rangle+2\sum_{i\in W}t_i\langle  Ax,Ae_i\rangle+2\sum_{i,j\in W,i\neq j}t_it_j\langle  Ae_i,Ae_j\rangle+\sum_{i\in W}t_i^2\langle  Ae_i,Ae_i\rangle\nonumber\\
 &&\ \ -\left(t^4\langle Ax,Ax\rangle-2t^2\sum_{i\in W}t_i\langle  Ax,Ae_i\rangle+2\sum_{i,j\in W,i\neq j}t_it_j\langle Ae_i,Ae_j\rangle+\sum_{i\in W}t_i^2\langle  Ae_i,Ae_i\rangle\right)\nonumber\\
 &&=(1-t^4)\langle Ax,Ax\rangle+2(1+t^2)\sum_{i\in W}t_i\langle  Ax,Ae_i\rangle\nonumber\\
 &&=(1-t^4)\left(\langle Ax,Ax\rangle-\frac{2}{1-t^2}\sum_{i\in W}|t_i||\langle  Ax,Ae_i\rangle|\right)\nonumber\\
 &&=(1-t^4)\left(\langle Ax,Ax\rangle-\frac{2}{1-t^2}(1-t^2)\frac{C}{2}\sum_{i\in W}|\langle  Ax,Ae_i\rangle|\right)\nonumber\\
 &&=(1-t^4)\left(\langle Ax,Ax\rangle-C\sum_{i\in W}|\langle  Ax,Ae_i\rangle|\right).\nonumber
 \end{eqnarray}
We have already completed the proof of the Lemma \ref{lem2}.
 \end{proof}
\textbf{Remark}\ \textbf{1.} The Lemma \ref{lem2} is a generalization of Lemma $\textrm{II}.1$ in \cite{16'}. The main idea of its proof is from the idea of Lemma $\textrm{II}.1$.

\textbf{Remark}\ \textbf{2.}
If $x$ is replaced by $cx$ with non-zero scalar $c$, the results of Lemma \ref{lem2} keep unchanged.

\section{ Main results }\label{3}

\subsection{Noiseless case}
\ \

It is well known that if at least one index of $N$ indices selected  is correct in every iteration, the gOMP makes a success, i.e., in each iteration, there exists  $\beta^{k}_1>\alpha_N^k\ (1\leq k\leq K)$.
The following theorems show a sufficient condition  guarantees the gOMP algorithm success. The proof of these theorems mainly uses Lemmas \ref{lem1} and \ref{lem2}.  By Remark $2$ we assume $\|x\|_2=1$ in the proof of Theorem \ref{the1} and $\|\tilde{\omega}_{T\cup\Lambda^k}\|_2=1$ in the proof of Theorem \ref{the2} .
\begin{thm}\label{the1}
Suppose  $x$ is a $K$-sparse signal and the restricted isometry constant $\delta_{K+N}$  of the sensing matrix $A$ satisfies
\begin{eqnarray}\label{e34}
\delta_{K+N}<\frac{1}{\sqrt{\frac{K}{N}+1}}.
\end{eqnarray}
Then the gOMP algorithm makes a success in the first iteration.
\end{thm}
\textbf{Remark}\ \textbf{3.}
In \cite{20}, authors proved that
\begin{eqnarray*}
\delta_{K+N} <\frac{\sqrt{N}}{\sqrt{K}+\sqrt{N}}
\end{eqnarray*}
is sufficient to make a success in the first iteration of the gOMP.
It is clear that
\begin{eqnarray*}
\delta_{K+N} <\frac{\sqrt{N}}{\sqrt{K}+\sqrt{N}}<\frac{1}{\sqrt{\frac{K}{N}+1}},
\end{eqnarray*}
i.e., the sufficient condition \eqref{e34} is weaker than that in \cite{20}.

\begin{proof}
In the first iteration, by the definition of $\alpha_N^1$, it satisfies
\begin{eqnarray}\label{e4}
\alpha_N^1&=&\min\{|\langle Ae_i,Ax\rangle||i\in W_1\}\nonumber\\
&\leq&\frac{\sum_{i\in W_1}|\langle Ae_i,Ax\rangle|}{N},
\end{eqnarray}
where $W_1\subseteq T^c$.

For  $\beta_1^1$ which  is the  largest correlation in magnitude in $A_T'Ax$, we have
\begin{eqnarray}\label{e5}
\langle Ax, Ax \rangle&=&\langle A\sum_{i\in T}x_ie_i, Ax \rangle \nonumber\\
&=&\sum_{i\in T}x_i\langle Ae_i, Ax \rangle\nonumber\\
&\leq &\sum_{i\in T}|x_i||\langle Ae_i, Ax \rangle|\nonumber\\
&\leq&\beta_1^1\|x\|_1\nonumber\\
&\leq&\beta_1^1\sqrt{K}\|x\|_2\nonumber\\
&\leq&\beta_1^1\sqrt{K}.
\end{eqnarray}

Let $t=-\frac{\sqrt{\frac{K}{N}+1}-1}{\sqrt{\frac{K}{N}}}$
and
\begin{eqnarray*}
 t_i=\left\{
       \begin{array}{ll}
         -\frac{\sqrt{K}}{2N}(1-t^2), & \hbox{$\langle Ax, Ae_i\rangle\geq 0$;} \\
         +\frac{\sqrt{K}}{2N}(1-t^2), & \hbox{$\langle Ax, Ae_i\rangle<0$,}
       \end{array}
     \right.
\end{eqnarray*}
where $i\in W_1\subseteq T^c$ with $|W_1|=N$, then we have that
\begin{eqnarray*}
t^2=\frac{\sqrt{\frac{K}{N}+1}-1}{\sqrt{\frac{K}{N}+1}+1}<1
\end{eqnarray*}
and
 \begin{eqnarray}\label{e33}
\sum_{i\in W_1}t_i^2&=&\left(\frac{\sqrt{K}}{2N}(1-t^2)\right)^2N\nonumber\\
&=&\frac{K}{4N}(1-t^2)^2\nonumber\\
&=&\frac{K}{4N}\left(1-\frac{\sqrt{\frac{K}{N}+1}-1}{\sqrt{\frac{K}{N}+1}+1}\right)^2\nonumber\\
&=&\frac{K}{N}\frac{1}{\left(\sqrt{\frac{K}{N}+1}+1\right)^2}\nonumber\\
&=&\frac{\sqrt{\frac{K}{N}+1}-1}{\sqrt{\frac{K}{N}+1}+1}=t^2.
 \end{eqnarray}
 By \eqref{e4}, \eqref{e5} and Lemma \ref{lem2}, we obtain
\begin{eqnarray}\label{e6}
(1-t^4)\sqrt{K}(\beta_1^1-\alpha_N^1)&\geq&(1-t^4)\left(\langle Ax, Ax \rangle-\sqrt{K}\frac{\sum_{i\in W_1}|\langle Ae_i,Ax\rangle|}{N}\right)\nonumber\\
&=&\|A(x+\sum_{i\in W_1}t_ie_i)\|_2^2-\|A(t^2x-\sum_{i\in W_1}t_ie_i)\|_2^2.
\end{eqnarray}
Because the sensing matrix $A$ satisfies  the RIP of order $K+N$
with $\delta_{K+N}$, $\|x\|_2=1$ with $supp(x)\subseteq T$, $W_1\subseteq T^c$, it follows from \eqref{e33} that
\begin{eqnarray*}\label{e7}
&&\|A(x+\sum_{i\in W_1}t_ie_i)\|_2^2-\|A(t^2x-\sum_{i\in W_1}t_ie_i)\|_2^2\nonumber\\
&&\geq(1-\delta_{K+N})\left(\|x+\sum_{i\in W_1}t_ie_i\|_2^2\right)-(1+\delta_{K+N})\left(\|t^2x-\sum_{i\in W_1}t_ie_i\|_2^2\right)\nonumber\\
&&=(1-\delta_{K+N})\left(\|x\|_2^2+\sum_{i\in W_1}t_i^2\right)-(1+\delta_{K+N})\left(t^4\|x\|_2^2+\sum_{i\in W_1}t_i^2\right)\nonumber\\
&&=(1-\delta_{K+N})(1+t^2)-(1+\delta_{K+N})(t^4+t^2)\nonumber\\
&&=(1-t^4)-\delta_{K+N}(1+t^2)^2\nonumber\\
&&=(1+t^2)^2\left(\frac{1-t^2}{1+t^2}-\delta_{K+N}\right).
\end{eqnarray*}
It follows from the definition of $t$ that
\begin{eqnarray*}\label{e8}
\frac{1-t^2}{1+t^2}&=&\frac{1-\frac{\sqrt{\frac{K}{N}+1}-1}{\sqrt{\frac{K}{N}+1}+1}}{1+\frac{\sqrt{\frac{K}{N}+1}-1}{\sqrt{\frac{K}{N}+1}+1}}\nonumber\\
&=&\frac{1}{\sqrt{\frac{K}{N}+1}}.
\end{eqnarray*}
Therefore by the condition $\delta_{K+N}<\frac{1}{\sqrt{\frac{K}{N}+1}}$, we obtain
\begin{eqnarray*}
 (1-t^4)\sqrt{K}(\beta_1^1-\alpha_N^1)
&\geq&(1+t^2)^2\left(\frac{1-t^2}{1+t^2}-\delta_{K+N}\right)\nonumber\\
&\geq&(1+t^2)^2\left(\frac{1}{\sqrt{\frac{K}{N}+1}}-\delta_{K+N}\right)\nonumber\\
&>&0,
\end{eqnarray*}
i.e., $\beta_1^1>\alpha_N^1$ which represents the gOMP selects  at least one index from the support $T$.

As mentioned, if $\delta_{K+N}<\frac{1}{\sqrt{\frac{K}{N}+1}}$, then the gOMP algorithm makes a success in the first iteration.
\end{proof}
\begin{thm}\label{the2}
If the gOMP algorithm  has performed $k$ iterations successfully, where $1\leq k<K$. And the sensing matrix $A$ satisfies the RIP of order $NK+1$ with RIC $\delta_{NK+1}$ fulfilling
\begin{eqnarray*}
\delta_{NK+1}<\frac{1}{\sqrt{\frac{K}{N}+1}}.
\end{eqnarray*}
Then in the $(k+1)$-th iteration, the gOMP will make a success.
\end{thm}
\begin{proof}
For the gOMP algorithm, $r^k=P^\perp_{\Lambda^k}y$ is orthogonal to each column of $A_{\Lambda^k}$ then
\begin{eqnarray*}
r^k&=&P^\perp_{\Lambda^k}y\\
&=&P^\perp_{\Lambda^k}A_Tx_T\\
&=&P^\perp_{\Lambda^k}(A_{T-{\Lambda^k}}x_{T-{\Lambda^k}}+A_{T\cap \Lambda^k}x_{T\cap\Lambda^k})\\
&=& P^\perp_{\Lambda^k}A_{T-{\Lambda^k}}x_{T-{\Lambda^k}}\\
&=&A_{T-{\Lambda^k}}x_{T-{\Lambda^k}}-P_{\Lambda^k}A_{T-{\Lambda^k}}x_{T-{\Lambda^k}}\\
&=&A_{T-{\Lambda^k}}x_{T-{\Lambda^k}}-A_{\Lambda^k}z_{\Lambda^k}\\
&=&A_{T\cup \Lambda^k}\omega_{T\cup \Lambda^k},
\end{eqnarray*}
where we used the fact that $P_{\Lambda^k}A_{T-{\Lambda^k}}x_{T-{\Lambda^k}}\in span(A_{\Lambda^k})$, so $P_{\Lambda^k}A_{T-{\Lambda^k}}x_{T-{\Lambda^k}}$ can be written as $A_{\Lambda^k}z_{\Lambda^k}$ for some $z_{\Lambda^k}\in \R^{|\Lambda^k|}$ and $\omega_{T\cup\Lambda^k}$ is given by
\begin{eqnarray*}
\omega_{T\cup \Lambda^k}=\left(
                              \begin{array}{c}
                                x_{T-\Lambda^k} \\
                                -z_{\Lambda^k}\\
                              \end{array}
                            \right).
\end{eqnarray*}

By the definition of $\alpha_N^{k+1}$ and $\beta_1^{k+1}$, we have that
\begin{eqnarray}\label{e9}
\alpha_N^{k+1}&=&\min\{|\langle Ae_i, r^k\rangle||i\in W_{k+1},W_{k+1} \subseteq (T\cup \Lambda^k)^c\}\nonumber\\
&\leq&\frac{\sum_{i\in W_{k+1}}|\langle Ae_i,r^k\rangle|}{N}\nonumber\\
&=&\frac{\sum_{i\in W_{k+1}}|\langle Ae_i, A_{T\cup\Lambda^k}\omega_{T\cup\Lambda^k}\rangle|}{N}\nonumber\\
&=&\frac{\sum_{i\in W_{k+1}}|\langle Ae_i, A\widetilde{\omega}_{T\cup\Lambda^k}\rangle|}{N}
\end{eqnarray}
and
\begin{eqnarray}\label{e10'}
\beta_1^{k+1}&=&\|A^{'}_{T-\Lambda^k}r^k\|_{\infty}\nonumber\\
&=&\|[A_{T-\Lambda^k}\ A_{T\cap\Lambda^k}]^{'}A_{T\cup\Lambda^k}\omega_{T\cup\Lambda^k}\|_{\infty}\nonumber\\
&=&\|A^{'}_TA_{T\cup\Lambda^k}\omega_{T\cup\Lambda^k}\|_{\infty}\nonumber\\
&=&\|[A_T\ A_{\Lambda^k-T}]^{'}A_{T\cup\Lambda^k}\omega_{T\cup\Lambda^k}\|_{\infty}\nonumber\\
&=&\|A^{'}_{T\cup \Lambda^k}A_{T\cup\Lambda^k}\omega_{T\cup\Lambda^k}\|_{\infty}.
\end{eqnarray}
Notice the fact that
\begin{eqnarray}\label{e16}
&&\|A^{'}_TA_{T\cup\Lambda^k}\omega_{T\cup\Lambda^k}\|_{\infty}\nonumber\\
&&\geq\frac{1}{\sqrt{K}}\|A^{'}_TA_{T\cup\Lambda^k}\omega_{T\cup\Lambda^k}\|_2\nonumber\\
&&=\frac{1}{\sqrt{K}}\|A_{T\cup\Lambda^k}^{'}A_{T\cup\Lambda^k}\omega_{T\cup\Lambda^k}\|_2.
\end{eqnarray}
By the hypothesis of  $\|\omega_{T\cup\Lambda^k}\|_2=1$, \eqref{e10'} and \eqref{e16},  it follows that
\begin{eqnarray}\label{e11}
\langle A\widetilde{\omega}_{T\cup \Lambda^k}, A\widetilde{\omega}_{T\cup \Lambda^k}\rangle&=&\langle A_{T\cup \Lambda^k}\omega_{T\cup \Lambda^k}, A_{T\cup \Lambda^k}\omega_{T\cup \Lambda^k}\rangle\nonumber\\
&=&\langle A_{T\cup \Lambda^k}^{'}A_{T\cup \Lambda^k}\omega_{T\cup \Lambda^k}, \omega_{T\cup \Lambda^k}\rangle \nonumber\\
&\leq& \|A_{T\cup \Lambda^k}^{'}A_{T\cup \Lambda^k}\omega_{T\cup \Lambda^k}\|_2\|\omega_{T\cup \Lambda^k}\|_2\nonumber\\
&\leq&\sqrt{K}\beta_1^{k+1}\|\omega_{T\cup \Lambda^k}\|_2\nonumber\\
&=&\sqrt{K}\beta_1^{k+1}.
\end{eqnarray}

As in the proof of Theorem \ref{the1},
let $t=-\frac{\sqrt{\frac{K}{N}+1}-1}{\sqrt{\frac{K}{N}}}$ and
\begin{eqnarray*}
 t_i=\left\{
       \begin{array}{ll}
         -\frac{\sqrt{K}}{2N}(1-t^2), & \hbox{$\langle Ax, Ae_i\rangle\geq 0$;} \\
         +\frac{\sqrt{K}}{2N}(1-t^2), & \hbox{$\langle Ax, Ae_i\rangle<0$,}
       \end{array}
     \right.
\end{eqnarray*}
where $i\in W_{k+1}\subseteq (\Lambda^k\cup T)^c$. By \eqref{e9}, \eqref{e11} and Lemma \ref{lem2}, we obtain
\begin{eqnarray}\label{e12}
&&(1-t^4)\sqrt{K}(\beta_1^{k+1}-\alpha_N^{k+1})\nonumber\\
&&\geq(1-t^4)\left(\langle A\widetilde{\omega}_{T\cup \Lambda^k}, A\widetilde{\omega}_{T\cup \Lambda^k}\rangle-\sqrt{K}\frac{\sum_{i\in W_{k+1}}|\langle Ae_i,A\widetilde{\omega}_{T\cup \Lambda^k}\rangle|}{N}\right)\nonumber\\
&&=\|A(\widetilde{\omega}_{T\cup \Lambda^k}+\sum_{i\in W_{k+1}}t_ie_i)\|_2^2-\|A(t^2\widetilde{\omega}_{T\cup \Lambda^k}-\sum_{i\in W_{k+1}}t_ie_i)\|_2^2.
\end{eqnarray}
Let $l=|T\cap\Lambda^k|$, then $k\leq l\leq K$ and $Nk+K-l+N\leq NK+1$.  Since $A$ satisfies  RIP of order $NK+1$
with $\delta_{NK+1}$, $\|\widetilde{\omega}_{T\cup \Lambda^k}\|_2=1$ with $supp(\widetilde{\omega}_{T\cup \Lambda^k})\subseteq T\cup \Lambda^k$, $W_{k+1}\subseteq (T\cup\Lambda^k)^c$, it follows from Lemma \ref{lem1} that
\begin{eqnarray*}\label{e7}
&&\|A(\widetilde{\omega}_{T\cup \Lambda^k}+\sum_{i\in W_{k+1}}t_ie_i)\|_2^2-\|A(t^2\widetilde{\omega}_{T\cup \Lambda^k}-\sum_{i\in W_{k+1}}t_ie_i)\|_2^2\nonumber\\
&&\geq(1-\delta_{Nk+K-l+N})\left(\|\widetilde{\omega}_{T\cup \Lambda^k}+\sum_{i\in W_{k+1}}t_ie_i\|_2^2\right)\nonumber\\
&&\ \ -(1+\delta_{Nk+K-l+N})\left(\|t^2\widetilde{\omega}_{T\cup \Lambda^k}-\sum_{i\in W_{k+1}}t_ie_i\|_2^2\right)\nonumber\\
&&=(1-\delta_{Nk+K-l+N})\left(\|\widetilde{\omega}_{T\cup \Lambda^k}\|_2^2+\sum_{i\in W_{k+1}}t_i^2\right)\nonumber\\
&&\ \ -(1+\delta_{Nk+K-l+N})\left(t^4\|\widetilde{\omega}_{T\cup \Lambda^k}\|_2^2+\sum_{i\in W_{k+1}}t_i^2\right)\nonumber\\
&&=(1-\delta_{Nk+K-l+N})(1+t^2)-(1+\delta_{Nk+K-l+N})(t^4+t^2)\nonumber\\
&&=(1+t^2)^2\left(\frac{1-t^2}{1+t^2}-\delta_{Nk+K-l+N}\right)\nonumber\\
&&\geq(1+t^2)^2\left(\frac{1-t^2}{1+t^2}-\delta_{NK+1}\right)\nonumber.
\end{eqnarray*}
Since
\begin{eqnarray*}\label{e8}
\frac{1-t^2}{1+t^2}=\frac{1}{\sqrt{\frac{K}{N}+1}}
\end{eqnarray*}
and  the condition $\delta_{NK+1}<\frac{1}{\sqrt{\frac{K}{N}+1}}$, we obtain
\begin{eqnarray*}
 (1-t^4)\sqrt{K}(\beta_1^{k+1}-\alpha_N^{k+1})
&\geq&(1+t^2)^2\left(\frac{1-t^2}{1+t^2}-\delta_{NK+1}\right)\nonumber\\
&\geq&(1+t^2)^2\left(\frac{1}{\sqrt{\frac{K}{N}+1}}-\delta_{NK+1}\right)\nonumber\\
&>&0,
\end{eqnarray*}
i.e., $\beta_1^{k+1}>\alpha_N^{k+1}$ which ensures that the set $\Lambda^{k+1}$ contains at least one correct index in the $(k+1)$-th iteration of the gOMP algorithm.

As mentioned, we have completed the proof of the theorem.
\end{proof}
Now combining the condition for success in the first iteration in Theorem \ref{the1} with that in non-initial iterations in Theorem \ref{the2}, we obtain overall sufficient condition of the gOMP algorithm guaranteeing the perfect recovery of $K$-sparse signals in the following theorem.
\begin{thm}\label{the3}
Suppose  $x$ is a $K$-sparse signal and the sensing matrix $A$ satisfies RIP of order $KN+1$ with the RIC $\delta_{NK+1}$ fulfilling
\begin{eqnarray*}
\delta_{NK+1}<\frac{1}{\sqrt{\frac{K}{N}+1}}.
\end{eqnarray*}
Then the gOMP algorithm can recover the signal $x$ exactly.
\end{thm}
\begin{proof}
For $N\geq1,\  K\geq1$ and $N\leq \min\{K,\ \frac{m}{K}\}$, then $K+N\leq NK+1$. It follows Lemma \ref{lem1} that
\begin{eqnarray*}
\delta_{K+N}\leq\delta_{NK+1}<\frac{1}{\sqrt{\frac{K}{N}+1}}.
\end{eqnarray*}
By Theorems \ref{the1} and \ref{the2},  the gOMP algorithm can recover perfectly any $K$-sparse signals  under the sufficient condition
$\delta_{NK+1}<\frac{1}{\sqrt{\frac{K}{N}+1}}$ from $y=Ax$.
\end{proof}
\textbf{Remark}\ \textbf{4.}
The condition $\delta_{NK+1}<\frac{1}{\sqrt{\frac{K}{N}+1}}$ is weaker than the sufficient condition $\delta_{NK+1}<\frac{\sqrt{N}}{\sqrt{K}+\sqrt{N}}$ in \cite{18}.

\textbf{Remark}\ \textbf{5.}
If $N=1$, this sufficient condition is consistent with the sharp condition $\delta_{K+1}<\frac{1}{\sqrt{K+1}}$  of OMP  in \cite{16'}.

In the following theorem, we show that the proposed bound $\delta_{NK+1}<\frac{1}{\sqrt{\frac{K}{N}+1}}$ is optimal.
\begin{thm}\label{the4}
For any given $K\in \N^+$, there are a $K$-sparse signal $x$ and a matrix $A$ satisfying
\begin{eqnarray*}
\delta_{NK+1}=\frac{1}{\sqrt{\frac{K}{N}+1}}
\end{eqnarray*}
such that the gOMP may fail.
\end{thm}
\begin{proof}
For any given positive integer $K$, let $A\in\R^{(NK+1)\times (NK+1)}$ be
\begin{eqnarray*}
A=
\begin{pmatrix}
  \ & \  & \  & 0 & \cdots & 0 & \frac{1}{b} & \cdots & \frac{1}{b} \\
  \  & \sqrt{\frac{K}{K+N}}I_{K} & \  & \vdots & \vdots & \vdots & \vdots & \vdots & \vdots\\
  \  & \ & \  & 0 & \cdots & 0 & \frac{1}{b} & \cdots & \frac{1}{b} \\
  0 & \cdots & 0 & \  & \  & \  & 0 & \cdots& 0 \\
  \vdots & \vdots & \vdots & \  & I_{NK+1-N-K} & \  & \vdots & \vdots & 0\\
  0 & \cdots & 0 & \  & \  & \  & 0 & \cdots & 0 \\
  \frac{1}{b} & \cdots & \frac{1}{b} & 0 & \cdots & 0  & \  & \  & \  \\
  \vdots & \vdots & \vdots & \vdots & \vdots & \vdots & \  & I_{N} & \  \\
  \frac{1}{b} & \cdots & \frac{1}{b}& 0 & \cdots & 0 & \  & \  & \  \\
\end{pmatrix},
\end{eqnarray*}
where $b=\sqrt{K(K+N)}$.
Then we have that
\begin{eqnarray*}
A'A=\begin{pmatrix}
  \ & \  & \  & 0 & \cdots & 0 & \frac{1}{K+N} & \cdots & \frac{1}{K+N} \\
  \  & \frac{K}{K+N}I_{K} & \  & \vdots & \vdots & \vdots & \vdots & \vdots & \vdots\\
  \  & \ & \  & 0 & \cdots & 0 & \frac{1}{K+N} & \cdots & \frac{1}{K+N} \\
  0 & \cdots & 0 & \  & \  & \  & 0 & \cdots& 0 \\
  \vdots & \vdots & \vdots & \  & I_{NK+1-N-K} & \  & \vdots & \vdots & 0\\
  0 & \cdots & 0 & \  & \  & \  & 0 & \cdots & 0 \\
  \frac{1}{K+N} & \cdots & \frac{1}{K+N} & 0 & \cdots & 0 & 1+\frac{1}{K+N} & \cdots &\frac{1}{K+N}\\
  \vdots & \vdots & \vdots & \vdots & \vdots & \vdots & \vdots& \ddots & \vdots \\
  \frac{1}{K+N} & \cdots & \frac{1}{K+N}& 0 & \cdots & 0 & \frac{1}{K+N}& \cdots&1+\frac{1}{K+N} \\
\end{pmatrix}.
\end{eqnarray*}
Moreover, by direct calculation, we obtain that
\begin{eqnarray*}
&&\begin{vmatrix}
  A'A-\lambda I \\
\end{vmatrix}=(1-\lambda)^{NK-K}(\frac{K}{K+N}-\lambda)^{K-1}(\lambda^2-2\lambda+\frac{K}{K+N}).
\end{eqnarray*}
It is  clear that $\frac{K}{K+N}$ and $1$ are eigenvalue of $A'A$ with multiplicity of $K-1$ and $NK-K$  respectively. $1\pm\frac{1}{\sqrt{\frac{K}{N}+1}}$ also are  eigenvalue of $A'A$.
 Therefore we have
 \begin{eqnarray*}
 \delta_{NK+1}=\frac{1}{\sqrt{\frac{K}{N}+1}}.
 \end{eqnarray*}

 Consider $K$-sparse signal $x=(1,1,\cdots,1,0\cdots,0)'\in\R^{NK+1}$, i.e., $T=supp(x)=\{1,2,\cdots,K\}$.
As $i\in T$, we have
 \begin{eqnarray*}
 |\langle Ae_i, y\rangle|&=&|\langle Ae_i, Ax\rangle|=|\langle A'Ae_i, x\rangle|=\frac{K}{K+N}.
 \end{eqnarray*}
For $i\in \{K+1, \cdots, NK+1-N\}$, it follows immediately that
\begin{eqnarray*}
 |\langle Ae_i, y\rangle|&=&|\langle Ae_i, Ax\rangle|
 =|\langle A'Ae_i, x\rangle|=0
\end{eqnarray*}
If $i\in \{NK+2-N, \cdots, NK+1\}$, we have
\begin{eqnarray*}
 |\langle Ae_i, y\rangle|&=&|\langle Ae_i, Ax\rangle|
 =|\langle A'Ae_i, x\rangle|
 =\frac{K}{K+N}.
 \end{eqnarray*}
Therefore, we have $\beta_1^1=\frac{K}{K+N}$ and $\alpha_N^1=\frac{K}{K+N}$  by the definitions of $\beta_1^1$ and $\alpha_N^1$, that is, $\beta_1^1=\alpha_N^1$.
 This implies the gOMP may fail to identify at least one correct index in the first iteration. So the gOMP algorithm may fail for the given matrix $A$ and the $K$-sparse signal $x$.
 \end{proof}
\subsection{Noise case}
\ \

In the subsection, we show a sufficient condition guarantees exact support identification by the gOMP algorithm from $y=Ax+e$.
 This sufficient condition is in terms of the RIC $\delta_{NK+1}$ and the minimum magnitude of the nonzero entries of $K$-sparse signal $x$. Here, we only consider $l_2$
 bounded noise, i.e., $\|e\|_2\leq\varepsilon$.

\begin{thm}\label{the5}
Suppose $\|e\|_2\leq\varepsilon$ and  the sensing matrix $A$ satisfies
\begin{eqnarray}\label{e13}
\delta_{NK+1}<\frac{1}{\sqrt{\frac{K}{N}+1}}.
\end{eqnarray}
Moreover, assume all the nonzero components $x_i$ satisfy
\begin{eqnarray}\label{e14}
|x_i|>\frac{\frac{2\sqrt{K}\varepsilon}{\sqrt{\frac{K}{N}+1}}}{\frac{1}{\sqrt{\frac{K}{N}+1}}-\delta_{NK+1}}.
\end{eqnarray}
Then  the gOMP algorithm  with the stopping rule $\|r^k\|_2\leq\varepsilon$ recovers the correct support of any $K$-sparse signals $x$.
\end{thm}
\begin{proof} Use mathematical induction method to prove the theorem. Suppose the gOMP performed $k$ iterations successfully.
Consider the $(k+1)$-th iteration. Firstly, we observe that
\begin{eqnarray*}
r^k&=&P^\perp_{\Lambda^k}y\\
&=&P^\perp_{\Lambda^k}A_Tx_T+P_{\Lambda^k}^\perp e\\
&=&A_{T\cup \Lambda^k}\omega_{T\cup \Lambda^k}+(I-P_{\Lambda^k})e
\end{eqnarray*}
for some $\omega_{T\cup \Lambda^k}$ as in the proof of Theorem \ref{the2}.
Consider the following two cases to prove the theorem.
\begin{itemize}
\item Case $1$: $T-\Lambda_k=\varnothing$
\end{itemize}

In this case, there is $T\subseteq\Lambda_k$.
Then  the correct support $T$ of the original $K$-sparse signal $x$ has already been  recovered.
\begin{itemize}
\item Case $2$: $T-\Lambda_k\neq\varnothing$, i.e., $|T-\Lambda_k|\geq1$
\end{itemize}

By the definitions of $\alpha^{k+1}_N$ and $\beta^{k+1}_1$, we obtain that
\begin{eqnarray}\label{e15}
\alpha_N^{k+1}&=&\min\{|\langle Ae_i, r^k\rangle||i\in W_{k+1},W_{k+1} \subseteq (T\cup \Lambda^k)^c\}\nonumber\\
&\leq&\frac{\sum_{i\in W_{k+1}}|\langle Ae_i,r^k\rangle|}{N}\nonumber\\
&\leq&\frac{\sum_{i\in W_{k+1}}|\langle Ae_i, A_{T\cup\Lambda^k}\omega_{T\cup\Lambda^k}\rangle|+\sum_{i\in W_{k+1}}|\langle Ae_i, (I-P_{\Lambda^k})e\rangle|}{N}\nonumber\\
&=&\frac{\sum_{i\in W_{k+1}}|\langle Ae_i, A\widetilde{\omega}_{T\cup\Lambda^k}\rangle|+\sum_{i\in W_{k+1}}|\langle Ae_i, (I-P_{\Lambda^k})e\rangle|}{N}
\end{eqnarray}
and
\begin{eqnarray}\label{e17}
\beta_1^{k+1}&=&\|A^{'}_{T-\Lambda^k}r^k\|_{\infty}\nonumber\\
&=&\|A^{'}_Tr^k\|_{\infty}\nonumber\\
&=&\|A^{'}_{T\cup \Lambda^k}r^k\|_{\infty}\nonumber\\
&\geq&\|A^{'}_{T\cup \Lambda^k}A_{T\cup\Lambda^k}\omega_{T\cup\Lambda^k}\|_{\infty}-\|A^{'}_{T\cup \Lambda^k}(I-P_{\Lambda^k})e\|_{\infty}.
\end{eqnarray}

Let $t=-\frac{\sqrt{\frac{K}{N}+1}-1}{\sqrt{\frac{K}{N}}}$
and
\begin{eqnarray*}
 t_i=\left\{
       \begin{array}{ll}
         -\frac{\sqrt{K}}{2N}(1-t^2)\|\omega_{T\cup \Lambda^k}\|_2, & \hbox{$\langle Ax, Ae_i\rangle\geq 0$;} \\
         +\frac{\sqrt{K}}{2N}(1-t^2)\|\omega_{T\cup \Lambda^k}\|_2, & \hbox{$\langle Ax, Ae_i\rangle<0$,}
       \end{array}
     \right.
\end{eqnarray*}
where $i\in W_{k+1}\subseteq (\Lambda^k\cup T)^c$.  Then we have
\begin{eqnarray}\label{e18}
\sum_{i\in W_{k+1}}t_i^2=t^2\|\omega_{T\cup \Lambda^k}\|_2^2.
\end{eqnarray}
It follows from \eqref{e16}, \eqref{e11}, \eqref{e15} and \eqref{e17} that
\begin{eqnarray*}
&&(1-t^4)\sqrt{K}\|\omega_{T\cup \Lambda^k}\|_2(\beta_1^{k+1}-\alpha_N^{k+1})\nonumber\\
&&\geq(1-t^4)\bigg(\langle A\widetilde{\omega}_{T\cup \Lambda^k}, A\widetilde{\omega}_{T\cup \Lambda^k} \rangle-\sqrt{K}\|\omega_{T\cup \Lambda^k}\|_2\|A^{'}_{T\cup \Lambda^k}(I-P_{\Lambda^k})e\|_{\infty}\bigg.\nonumber\\
&&\ \ \left.-\frac{\sqrt{K}\|\omega_{T\cup \Lambda^k}\|_2(\sum_{i\in W_{k+1}}|\langle Ae_i,A\widetilde{\omega}_{T\cup \Lambda^k}\rangle|+\sum_{i\in W_{k+1}}|\langle Ae_i, (I-P_{\Lambda^k})e\rangle|)}{N}\right)\nonumber\\
&&=\|A(\widetilde{\omega}_{T\cup \Lambda^k}+\sum_{i\in W_{k+1}}t_ie_i)\|_2^2-\|A(t^2\widetilde{\omega}_{T\cup \Lambda^k}-\sum_{i\in W_{k+1}}t_ie_i)\|_2^2\nonumber\\
&&\ \ -(1-t^4)\sqrt{K}\|\omega_{T\cup \Lambda^k}\|_2\left(\|A^{'}_{T\cup \Lambda^k}(I-P_{\Lambda^k})e\|_{\infty}
+\frac{\sum_{i\in W_{k+1}}|\langle Ae_i, (I-P_{\Lambda^k})e\rangle|}{N}\right).
\end{eqnarray*}
As in the proof of Theorem \ref{the2}, $l=|T\cap\Lambda^k|$ then $Nk+K-l+N\leq NK+1$.
Because  $A$ satisfies RIP of order $NK+1$
with $\delta_{Nk+1}$, $supp(\widetilde{\omega}_{T\cup \Lambda^k})\subseteq T\cup \Lambda^k$, $W_{k+1}\subseteq (T\cup\Lambda^k)^c$, it follows from \eqref{e18} and Lemma \ref{lem1} that
\begin{eqnarray*}\label{e7}
&&\|A(\widetilde{\omega}_{T\cup \Lambda^k}+\sum_{i\in W_{k+1}}t_ie_i)\|_2^2-\|A(t^2\widetilde{\omega}_{T\cup \Lambda^k}-\sum_{i\in W_{k+1}}t_ie_i)\|_2^2\nonumber\\
&&\geq(1-\delta_{Nk+K-l+N})\left(\|\widetilde{\omega}_{T\cup \Lambda^k}+\sum_{i\in W_{k+1}}t_ie_i\|_2^2\right)\nonumber\\
&&\ \ -(1+\delta_{Nk+K-l+N})\left(\|t^2\widetilde{\omega}_{T\cup \Lambda^k}-\sum_{i\in W_{k+1}}t_ie_i\|_2^2\right)\nonumber\\
&&=(1-\delta_{Nk+K-l+N})\left(\|\widetilde{\omega}_{T\cup \Lambda^k}\|_2^2+\sum_{i\in W_{k+1}}t_i^2\right)\nonumber\\
&&\ \ -(1+\delta_{Nk+K-l+N})\left(t^4\|\widetilde{\omega}_{T\cup \Lambda^k}\|_2^2+\sum_{i\in W_{k+1}}t_i^2\right)\nonumber\\
&&=(1-\delta_{Nk+K-l+N})\|\widetilde{\omega}_{T\cup \Lambda^k}\|_2^2(1+t^2)-(1+\delta_{Nk+K-l+N})\|\widetilde{\omega}_{T\cup \Lambda^k}\|_2^2(t^4+t^2)\nonumber\\
&&=(1-t^4)\|\widetilde{\omega}_{T\cup \Lambda^k}\|_2^2-\delta_{Nk+K-l+N}\|\widetilde{\omega}_{T\cup \Lambda^k}\|_2^2(1+t^2)^2\nonumber\\
&&=(1+t^2)^2\|\widetilde{\omega}_{T\cup \Lambda^k}\|_2^2\left(\frac{1-t^2}{1+t^2}-\delta_{Nk+K-l+N}\right)\nonumber\\
&&\geq(1+t^2)^2\|\widetilde{\omega}_{T\cup \Lambda^k}\|_2^2\left(\frac{1-t^2}{1+t^2}-\delta_{NK+1}\right)\nonumber.
\end{eqnarray*}
Moreover, notice the fact that
\begin{eqnarray*}
\|A^{'}(I-P_{\Lambda^k})e\|_\infty=\max_{i}|\langle Ae_i,(I-P_{\Lambda^k})e\rangle|\leq\|Ae_i\|_2\|(I-P_{\Lambda^k})e\|_2\leq\|e\|_2\leq\varepsilon.
\end{eqnarray*}
By the above three inequalities, \eqref{e13} and \eqref{e14}, it follows that
\begin{eqnarray*}
&&(1-t^4)\sqrt{K}\|\omega_{T\cup \Lambda^k}\|_2(\beta_1^{k+1}-\alpha_N^{k+1})\nonumber\\
&&\geq(1+t^2)^2\|\widetilde{\omega}_{T\cup \Lambda^k}\|_2^2\left(\frac{1-t^2}{1+t^2}-\delta_{NK+1}\right)-(1-t^4)\sqrt{K}\|\omega_{T\cup \Lambda^k}\|_2\nonumber\\
&&\ \ \left(\|A^{'}_{T\cup \Lambda^k}(I-P_{\Lambda^k})e\|_{\infty}
+\frac{\sum_{i\in W_{k+1}}|\langle Ae_i, (I-P_{\Lambda^k})e\rangle|}{N}\right)\nonumber\\
&&\geq(1+t^2)^2\|\widetilde{\omega}_{T\cup \Lambda^k}\|_2^2\left(\frac{1-t^2}{1+t^2}-\delta_{NK+1}\right)-(1-t^4)\sqrt{K}\|\omega_{T\cup \Lambda^k}\|_2\left(\varepsilon+\frac{N\varepsilon}{N}\right)\nonumber\\
&&=(1+t^2)^2\|\widetilde{\omega}_{T\cup \Lambda^k}\|_2\left(\left(\frac{1-t^2}{1+t^2}-\delta_{NK+1}\right)\|\widetilde{\omega}_{T\cup \Lambda^k}\|_2-2\sqrt{K}\varepsilon\frac{1-t^2}{1+t^2}\right)\nonumber\\
&&\geq(1+t^2)^2\|\widetilde{\omega}_{T\cup \Lambda^k}\|_2\left(\left(\frac{1-t^2}{1+t^2}-\delta_{NK+1}\right)\|x_{T-\Lambda_k}\|_2-2\sqrt{K}\varepsilon\frac{1-t^2}{1+t^2}\right)\nonumber\\
&&\geq(1+t^2)^2\|\widetilde{\omega}_{T\cup \Lambda^k}\|_2\left(\left(\frac{1}{\sqrt{\frac{K}{N}+1}}-\delta_{NK+1}\right)|T-\Lambda_k|\min_{i\in T-\Lambda_k} |x_i|-\frac{2\sqrt{K}\varepsilon}{\sqrt{\frac{K}{N}+1}}\right)\nonumber\\
&&>0,
\end{eqnarray*}
i.e., $\beta_1^{k+1}>\alpha_N^{k+1}$ which guarantees at least one  index selected  from the correct support in the $(k+1)-$th iteration.

Therefore the  gOMP with the stopping rule $\|r^k\|_2\leq\varepsilon$ recovers the correct support of any $K$-sparse signal $x$ under conditions \eqref{e13} and \eqref{e14}.
\end{proof}
\section{Acknowledgements}
This work was supported by the NSF of China (Nos.11271050,
11371183).

\end{document}